\documentclass{ws-ijbc}

\usepackage{graphicx}
\graphicspath{{figures-pdf/}}
\usepackage{multirow,bigstrut}
\usepackage{amsmath,amssymb,amsfonts,amssymb,bm,float}
\usepackage{graphicx, color, xcolor}
\usepackage{epstopdf}
\usepackage{url}
\usepackage[bookmarks=false]{hyperref}

\newtheorem{property}{Property}
\newtheorem{Conjecture}{Conjecture}
\newcommand{\Zp}[1]{\mathbb{Z}_{p^{#1}}}
\newcommand{\Zt}[1]{\mathbb{Z}_{2^{#1}}}
\newcommand{\Zy}[1]{\mathbb{Z}_{3^{#1}}}

\begin{document}

\catchline{}{}{}{}{} 
\markboth{X. Lu, et al.}{Periodicity Analysis of the Logistic Map over Ring $\Zy{n}$}

\title{Periodicity Analysis of the Logistic Map over Ring \Large{$\Zy{n}$}}

\author{Xiaoxiong Lu}
\address{School of Mathematics and Computational Science,\\
Xiangtan University, Xiangtan 411105, Hunan, China}

\author{Eric Yong Xie}
\address{School of Computer Science,\\
Xiangtan University, Xiangtan 411105, Hunan, China}

\author{Chengqing Li\thanks{Corresponding author. Email: DrChengqingLi@gmail.com.}}

\address{Key Laboratory of Intelligent Computing \textup{\&} Information Processing of Ministry of Education,\\
 Xiangtan University, Xiangtan 411105, Hunan, China}

\maketitle

\begin{history}
Feb 16, 2023
\end{history}

\begin{abstract}
Periodicity analysis of sequences generated by a deterministic system is a long-standing challenge in both theoretical research and engineering applications.
To overcome the inevitable degradation of the Logistic map on a finite-precision circuit, its numerical domain is commonly converted from a real number field to a ring or a finite field.
This paper studies the period of sequences generated by iterating the Logistic map over ring $\Zy{n}$ from the perspective of its associate functional network, where every number in the ring is considered as a node, and the existing mapping relation between any two nodes is regarded as a directed edge.
The complete explicit form of the period of the sequences starting from any initial value is given theoretically and verified experimentally.
Moreover, conditions on the control parameter and initial value are derived, ensuring the corresponding sequences to achieve the maximum period over the ring.
The results can be used as ground truth for dynamical analysis
and cryptographical applications of the Logistic map over various domains.
\end{abstract}


\section{Introduction}

Complex dynamics of chaos systems attracted researchers use them as an alternative way to design secure and efficient pseudo-random generators and
encryption algorithms: Logistic map \cite{Collin:PRN:CP1992,Chen:Logistic:TCASII10,garcia2018chaos:TIM18,Bus:shiftlog:IJBC2023}, Chebyshev polynomials \cite{Liao:ITC:2010,Yoshioka:ChebyshevPk:TCAS2:2020}, R\'enyi map~\cite{Addabbo:Reyi:TCSI2007}, Tent map \cite{Jessa:tent:CSI2002}, Cat map \cite{Falcioni:PRNS:PRE2005, Catchen2013period2e,Souza:cat3m:CSII2021}, H\'{e}non map \cite{Galias:Henon:TCASI22}, Chua's attractor \cite{cqli:Diode:TCASI19},
and Lorenz system~\cite{zhoulili:encrypt:sc22}.
Among them, the Logistic map is one of the simplest systems exhibiting complex dynamics, and it is often used as a classic case to illustrate how complex chaotic phenomena arise from simple models~\cite{May:property:Nature76,cqli:network:TCASI2019,cqli:block:JISA20}.
Due to the effect of rounding errors and limited presentation precision in any digital device, real implementation of a chaotic system inevitably leads to tiresome dynamics degradation problem~\cite{kocarev2006discrete,cqli:network:TCASI2019}.
To solve this problem, the numerical domain of the chaotic map was
suggested by some researchers to extend from real number field to residue ring or finite field.
For example, some public-key encryption algorithms based on Chebyshev polynomials were implemented over ring or finite field to improve security performance and reliability~\cite{Yoshioka:ChebyshevPk:TCAS2:2020}.
In general, the sequences generated by iterating the Logistic map over ring or finite field have better randomness than that obtained over real number field, which allures them to design more seemingly efficient image encryption algorithms~\cite{Kazuyoshi:Period:IDICE17,Yang:log:Scis2017,Yang:log:Scis2018,Li:log3graph:IJPRA2019,Yoshida:log2n:ISIT2014}.

Period distribution of the sequences generated by a chaotic map over a given domain is a fundamental characteristic for evaluating its performance, which serves as precondition for real measurement of
the corresponding application merits~\cite{cqli:DNA:JVCI22}.
Using the generating function and Hensel's lifting method, F. Chen et al. systematically analyzed period distribution of the sequences generated by iterating Chebyshev polynomial over a prime field and/or Cat map over ring $\Zp{n}$, where $p$ is a prime number~\cite{Liao:ITC:2010,chen2012periodpe,Catchen2013period2e}.
Generally, the generating function method is used to deal with the period of sequences generated by a linear recursive generator.
However, the nonlinear complexity of sequences generated by iterating the Logistic map may make the approach fail.
Alternatively, the \emph{state-mapping network} (SMN), also known as functional graph in some research fields, is an essential visible way to analyze period of sequences~\cite{Roger:sqaure:DM1996,Vasiga:sqaure:DM2004}. A periodic sequence can be viewed as a circle in a SMN.
In~\cite{Roger:sqaure:DM1996,Vasiga:sqaure:DM2004}, the associated functional graphs of functions $x^2+c$ and $x^2$ over prime field are disclosed.
Reference~\cite{Yoshida:log2n:ISIT2014} presented some statistical properties and conjectures about the maximum period of sequences generated by iterating the Logistic map over $\Zt{n}$.
Using dynatomic polynomials, Yang~\emph{et al.} proved conjectures given in \cite{Yoshida:log2n:ISIT2014} and analyzed the maximum period of such sequences over $\Zy{n}$ for different control parameters~\cite{Yang:log:Scis2017,Yang:log:Scis2018}.
However, reference~\cite{Li:log3graph:IJPRA2019} pointed out some errors about
some reported properties in \cite{Yang:log:Scis2017} and summarized a conjecture about the maximum period of sequences over $\Zy{n}$.

Multiple research groups studied the sequences generated by iterating the Logistic map over ring $\Zy{n}$ from various perspectives
~\cite{Yang:log:Scis2017,Yang:log:Scis2018,Li:log3graph:IJPRA2019,Yoshida:log2n:ISIT2014}.
However, the full information on period of the sequence for any parameter and initial value is
still unclear. And the condition for the generated sequences owning the maximum period is also unknown, which limits its applications in image encryption and other cryptographic applications.
Using the internal structure of the SMN of the Logistic map over ring $\Zy{n}$,
this paper studied the period of sequences generated by iterating the map from any initial value.
The conjecture given in \cite{Li:log3graph:IJPRA2019} and some theorems given in \cite{Yang:log:Scis2017} are revised and proved.

The rest of the paper is organized as follows. Section~\ref{sec:period} reviews the known results on the Logistic map over ring $\Zy{n}$, and discloses the corresponding period distribution.
Finally, some conclusions are drawn.

\section{Sequences generated by iterating the Logistic map over ring $\Zy{n}$}
\label{sec:period}

In this section, we first review some previous work on the period of sequences generated by iterating the Logistic map over ring $\Zy{n}$. Then some general properties of the Logistic map are given. Finally, the explicit expression of the period of sequence generated by iterating the Logistic map from any initial value in ring $\Zy{n}$ is disclosed.

\subsection{Preliminary}

Let $\Zp{n}=\{0, 1, 2, \cdots, p^n-1\}$ be the ring of residue classes modulo $p^n$ with respect to modular addition and multiplication, where $p$ is a prime number and $n$ is a natural number.
In \cite{Kazuyoshi:Period:IDICE17,Yang:log:Scis2017,Yang:log:Scis2018,Li:log3graph:IJPRA2019, Yoshida:log2n:ISIT2014},
the numerical domain of the Logistic map was extended from real number field to ring $\Zp{n}$ with different $p$ and $n$. It can be expressed as
\begin{equation}\label{eq:log:ring}
\begin{aligned}
f_{p^n}(x) &= \frac{\mu x(p^n-1-x)}{p^n-1} \bmod {p^n}\\
&= \mu x(x+1) \bmod {p^n},
\end{aligned}
\end{equation}
where $\mu, x\in \Zp{n}$.
Given an initial value $x_0\in\Zp{n}$, the $i$-th iteration of the Logistic map over ring $\Zp{n}$ is
\begin{equation}
\label{eq:sequence}
x_i= f^i_{p^n}(x_0)= f_{p^n}(f_{p^n}^{i-1}(x_0)),
\end{equation}
where $i\geq 1$ and $f_{p^n}^0(x_0)=x_0$.
Then we can get a sequence $\{f^i_{p^n}(x_0)\}_{i\ge 0}$, which is denoted as $S(x_0; \mu, p^n)$.
If there exist integers $L>0$ and $m_0\geq 0$ such that
$x_{m+L}= x_m$ for all $m\geq m_0$, then sequence $S(x_0; \mu, p^n)$ is called \textit{ultimately periodic}, and $m_0$ is the \textit{pre-period} of the sequence.
Specially, the sequence is called \textit{periodic} if $m_0=0$.
The minimum positive integer among all possible values of $L$ is called the \textit{period} of the sequence,
which is denoted as $L(x_0; \mu, p^n)$.
Let
\begin{equation}
\label{eq:define:lmax }
L(\mu,p^n)= \max \{L(x_0; \mu,p^n) \mid x_0\in \Zp{n}\}
\end{equation}
represent the maximum period of sequences generated by iterating the Logistic map over ring $\Zp{n}$.
As for any $\mu$, reference~\cite{Yang:log:Scis2017} gave representation form of $L(\mu, 3^n)$:
\begin{equation}\label{eq:Lmax:3n}
L(\mu, 3^n)=
\begin{cases}
1 & \mbox{if } \mu\bmod 3 \in\{0,2\}; \\
3^{n-2} & \mbox{if } \mu\bmod 9=1;  \\
3^{n-3} & \mbox{if } \mu\bmod 9\in\{4,7\}.
\end{cases}
\end{equation}
However, reference~\cite{Li:log3graph:IJPRA2019} stated that Eq.~\eqref{eq:Lmax:3n}
does not hold for case $\mu\bmod 9\in\{1, 4, 7\}$ and concluded Conjecture~\ref{conj:3n}.
Moreover, we find Eq.~\eqref{eq:Lmax:3n} does not hold also when $\mu\bmod 3=2$. In Sec.~\ref{subsec:3n}, we revise Eq.~\eqref{eq:Lmax:3n} and prove Conjecture~\ref{conj:3n} as Corollary~\ref{Cor:x_0:max3n}.

\begin{Conjecture}
\label{conj:3n}
When $\mu\bmod 3=1$, the maximum period of sequences generated by iterating the Logistic map over ring $\Zy{n}$ is $3^{n-2}$.
\end{Conjecture}

Define $\mathbb{H}_{p^n}=\{x\mid x\bmod p=0, x\in\Zp{n}\}$ and
\[
F(x) = \mu x(x+1).
\]
We introduce some properties about $F(x)$, as shown in Properties~\ref{pro:f:bij}, \ref{property:a_1:a_2}, \ref{property:x^n}, and \ref{property:f2_3t}, which are useful for analyzing the explicit expression of the period of sequence generated by iterating the Logistic map over ring $\Zy{n}$.

\begin{property}
\label{pro:f:bij}
Define map $\Gamma:\mathbb{H}_{p^n} \rightarrow \mathbb{H}_{p^n}$ by $\Gamma(x) = \mu x(x+1)\bmod p^n$, then $\Gamma$ is bijective, where $\mu\bmod p\neq 0$.
\end{property}
\begin{proof}
First, let's prove that $\Gamma$ is injective.
Suppose $x',x''\in \mathbb{H}_{p^n}$ and $x'\neq x''$, one has
\begin{align*}
\Gamma(x')-\Gamma(x'')& = (\mu x'(x'+1)-\mu x''(x''+1))\bmod p^n\\
&=\mu(x'-x'')(x'+x''+1)\bmod p^n.
\end{align*}
Since $x',x'' \in \mathbb{H}_{p^n}$, $(x'+x''+1)\bmod p=1$, it follows from the above equation and $\mu\bmod p\neq 0$ that $\Gamma(x')\neq \Gamma(x'')$.
Hence, $\Gamma$ is injective.
Then, let's prove that $\Gamma$ is surjective, namely there exists $x'\in \mathbb{H}_{p^n}$ such that $\Gamma(x')=y'$ for any $y'\in \mathbb{H}_{p^n}$.
Operating polynomial $G(x) = \mu(x+1)x- y'$ over $\Zp{n}$, one has
$G(x)\equiv \bar{\mu}x(x+1)\pmod p$, where
$\bar{\mu}=\mu\bmod p\neq 0$.
Referring to Hensel's Lemma~\cite[Lemma 13.6]{Wan:Lecture:2003}, there exist $f_1(x)$ and $f_2(x)$ such that $G(x)= f_1(x)f_2(x)$ and $f_1(x)\equiv x\pmod p$, $f_2(x)\equiv \bar{\mu}(x+1)\pmod p$. It means
there exists $x'\bmod p=0$ such that $f_1(x) =x-x'$.
Thus, $G(x')=0$ and $\Gamma(x') =y'$.
Hence $\Gamma$ is surjective.
\end{proof}

\begin{property}
\label{property:a_1:a_2}
Function $F(x)$ satisfies
\[
F^n(x)= \sum^{2^n}_{i=3}a_{i,n}x^i+ \sum^{2n-1}_{i=n}\mu^ix^2+\mu^nx,
\]
where $n\geq 1$, $F^n(x)=F^{n-1}(F(x))$, $\mu$ and $a_{i,n}$ are positive integers.
\end{property}
\begin{proof}
Prove this property via mathematical induction on $n$. When $n=1$, $F(x)=\mu x^2+\mu x=\mu x(x+1)$, which means this property holds for $n=1$.
Assume that this property holds for $n=s$, namely
$F^s(x) = \sum^{2^s}_{i=3}a_{i,s}x^i + \sum^{2s-1}_{j=s}\mu^jx^2+\mu^sx.$
When $n=s+1$, one can get
\begin{align*}
F^{s+1}(x) &= F(F^s(x)) \\
&= \sum^{2^{s+1}}_{i=3}a_{i,s+1}x^i +\mu(\mu^{2s}x^2 +\sum^{2s-1}_{j=s}\mu^jx^2+\mu^sx)\\
&= \sum^{2^{s+1}}_{i=3}a_{i,s+1}x^i + \sum^{2s+1}_{j=s+1}\mu^jx^2+\mu^{s+1}x,
\end{align*}
where $a_{i,s+1}$ is a positive integer.
It yields that this property holds for $n=s+1$.
The above induction completes the proof of this property.
\end{proof}

\begin{property}\label{property:x^n}
As for any $x\in \mathbb{H}_{p^n}$ and $n\geq 3$, then
\[(x+k\cdot p^w)^n \equiv x^n \pmod {p^{w+2}},\]
where $k$ and $w$ are positive integers.
\end{property}
\begin{proof}
Since $x\in \mathbb{H}_{p^n}$, $x=b\cdot p$, and $b\in \{0,1,2\cdots, p^{n-1}-1\}$.
It yields that
$(x+k\cdot p^w)^n
=x^n+\sum^{n}_{i=1}\binom{n}{i}b^{n-i}\cdot k^i\cdot p^{n+(w-1)i}.$
As $n\geq 3$ and $i\geq 1$, $n+(w-1)i\geq w+2$ for any $w$.
Thus, $(x+k\cdot p^w)^n \equiv x^n \pmod {p^{w+2}}.$
\end{proof}

\begin{property}
\label{property:f2_3t}
If there is an integer $x\in \mathbb{H}_{p^n}$ satisfying
\begin{equation}
\label{eq:2w:change}
\begin{cases}
F^n(x)\equiv x \pmod{p^w};\\
F^n(x)\not\equiv x \pmod{p^{w+1}},
\end{cases}
\end{equation}
then
\[
F^{i\cdot n}(x) \equiv x+k\cdot p^w
\sum^{i-1}_{j=0}\mu^{jn}
 \pmod {p^{w+2}},
\]
where $n\geq 1$, $w\geq 2$, $k\bmod p\neq0$, and $(\sum^{n-1}_{j=0}\mu^j)\bmod p=0$.
\end{property}
\begin{proof}
Prove this property via mathematical induction on $i$.
When $i=1$, one has
$F^n(x)=x+k\cdot p^w \equiv x+k\cdot p^w \pmod {p^{w+2}}$
from relation~\eqref{eq:2w:change}.
So this property holds for $i=1$. Suppose that this property holds for $i=s$, namely
\[
F^{s\cdot n}(x) \equiv x+k\cdot p^w\sum^{s-1}_{j=0}\mu^{jn} \pmod {p^{w+2}}.
\]
When $i=s+1$, 
from the above congruence and Properties~\ref{property:a_1:a_2},~\ref{property:x^n}, one has
\begin{align*}
F^{(s+1)\cdot n}(x)
&=F^n(F^{s\cdot n}(x))\\
&\equiv \sum^{2^n}_{i=3}a_{i,n}(x+B)^i +\sum^{2n-1}_{j=n}\mu^j(x+B)^2+\mu^{n}(x+B)\pmod{p^{w+2}}\\
&\equiv \sum^{2^n}_{i=3}a_{i,n}x^i +\sum^{2n-1}_{j=n}\mu^jx^2
+2xB\sum^{2n-1}_{j=n}\mu^j+\mu^n(x+B)\pmod{p^{w+2}}\\
&\equiv F^n(x)+ 2xB\cdot u^n\sum^{n-1}_{j=0}\mu^j+\mu^nB \pmod{p^{w+2}},
\end{align*}
where $B=k\cdot p^w\sum^{s-1}_{j=0}\mu^{jn}$. 
Substituting $F^n(x)=x+k\cdot p^w$ and $(\sum^{n-1}_{j=0}\mu^j)\bmod p=0$ into the above congruence, one can get 
\[
F^{(s+1)\cdot n}(x)\equiv F^n(x)+ \mu^nB \pmod{p^{w+2}}\equiv x+k\cdot p^w\sum^{s}_{j=0}\mu^{jn}\pmod {p^{w+2}}.
\]
Thus, this property holds for $i=s+1$.
The above induction completes the proof of this property.
\end{proof}

\subsection{Explicit expression of the period of $S(x_0; \mu, 3^n)$}
\label{subsec:3n}

Let $F_{p^n}$ denote the associate SMN of the Logistic map over ring $\Zp{n}$.
It is constructed as follows: the $p^n$ numbers in ring $\Zp{n}$ are
separately considered as $p^n$ nodes; node $x$ is directly linked to node $y$ if and only if $y=f_{p^n}(x)$ \cite{cqli:network:TCASI2019}.

As a typical example, we draw $F_{3^n}$ with $\mu=19$ and $n=1, 2, 3, 4$ in Fig.~\ref{fig:SMN:3n}, which indicates some general rules:
any node satisfying $x\bmod 3=1$ is directly linked to a node satisfying $x\bmod 3=2$;
any node satisfying $x\bmod 3=2$ is directly linked to a node satisfying $x\bmod 3=0$;
all nodes satisfying $x\bmod 3=0$ in $F_{3^n}$ form some directed cycles for arbitrary parameter $n$.
Such rules are summarized in Property~\ref{pro:F3n}.
In addition, as for a directed cycle $C_n$ in $F_{3^n}$, if the length of $C_n$ is larger than three, $C_n$ is expanded to one cycle of length $3T_c$ in $F_{3^{n+1}}$.
For example, a cycle $``3\rightarrow12\rightarrow21"$ shown in Fig.~\ref{fig:SMN:3n}c)
is expanded to cycle of length nine $``3\rightarrow 66\rightarrow 21 \rightarrow 30 \rightarrow \cdots \rightarrow 3"$ shown in Fig.~\ref{fig:SMN:3n}d).
Moreover, Lemma~\ref{lemma:3:change} gives the condition that the length of cycle increases by three times with increase of parameter $n$.
Finally, combining Property~\ref{pro:F3n} and Lemma~\ref{lemma:3:change},
one can get explicit expression of the period of the Logistic map over $\Zy{n}$
as Theorem~\ref{thero:3n:period}.

\newlength\FourImW
\setlength\FourImW{0.1\columnwidth}
\newlength\twofigwidth
\setlength\twofigwidth{0.5\columnwidth}

\begin{figure}[!htb]
\centering
\begin{minipage}[c]{0.4\twofigwidth}
    \centering
    \includegraphics[width=\FourImW]{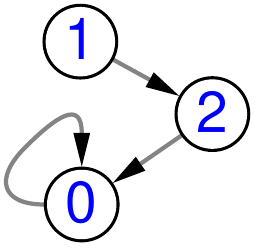}
    \\a)
    \vspace{0.4cm}
    \centering
    \includegraphics[width=2\FourImW]{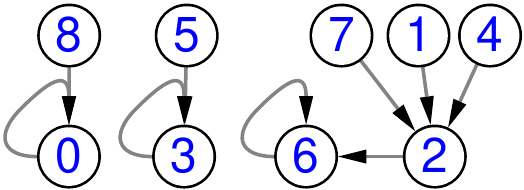}
    \vspace{0.3cm}
    b)\\
    \centering
    \includegraphics[width=4\FourImW]{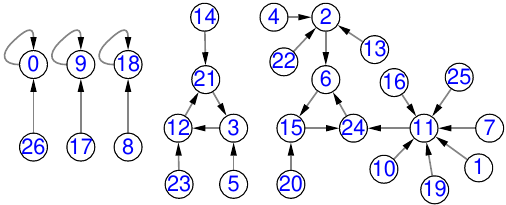}
    c)
\end{minipage}\hspace{3.4cm}
\begin{minipage}[c]{1\twofigwidth}
    \centering
    \includegraphics[width=1.1\twofigwidth]{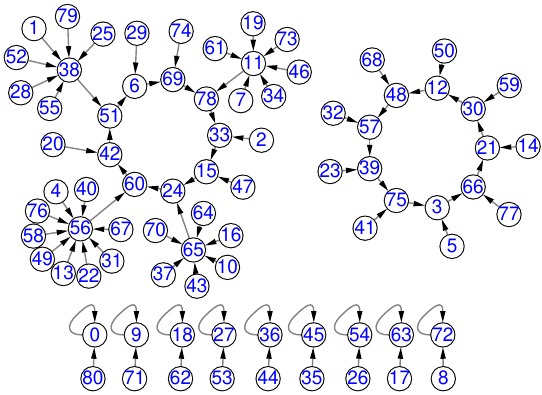}
    d)
\end{minipage}
\caption{Functional graphs of the Logistic map implemented with
\href{https://www.xilinx.com/products/design-tools/vivado.html}{Xilinx Vivado}
on various domains:
a) $\Zy{1}$; b) $\Zy{2}$; c) $\Zy{3}$; d) $\Zy{4}$.}
\label{fig:SMN:3n}
\end{figure}


\begin{property}
\label{pro:F3n}
As for any $x_0\in \Zy{n}$, sequence $S(x_0; \mu, 3^n)$ is periodic if $x_0\bmod 3=0$; ultimately periodic with pre-period $m$ otherwise, where $m=1$ when $x_0\bmod 3=2$ and $m=2$ when $x_0\bmod 3=1$.
\end{property}

\begin{proof}
Referring to Property~\ref{pro:f:bij}, one can get
$\Gamma(x)=\mu x(x+1)\bmod 3^n$ is a bijective function from $\mathbb{H}_{3^n}$ to itself.
Thus, if $x_0\bmod 3=0$, namely $x_0\in\mathbb{H}_{3^n}$, one has sequence $S(x_0; \mu, 3^n)$ is periodic from the definition of the sequence and \cite[Theorem 5.1.1]{hall1959marshall}.
If $x_0\bmod 3\neq 0$, then $x_m= f^m_{3^n}(x_0)\equiv 0 \pmod3$ from Eq.~\eqref{eq:log:ring}.
It means $x_m\in \mathbb{H}_{3^n}$ and sequence $S(x_m; \mu, 3^n)$ is periodic. Thus, sequence $S(x_0; \mu, 3^n)$
is ultimately periodic and its pre-period is $m$.
\end{proof}

\begin{lemma}\label{lemma:3:change}
If there is an integer $x\in \mathbb{H}_{3^n}$ satisfying
\begin{equation}\label{eq:3:change}
\begin{cases}
F^{\bar{\mu}}(x)\equiv x \pmod{3^w}; \\
F^{\bar{\mu}}(x)\not\equiv x \pmod{3^{w+1}},
\end{cases}
\end{equation}
then
\begin{equation}\label{eq:3t:change}
\begin{cases}
F^{\bar{\mu}\cdot 3^{t}}(x)\equiv x \pmod{3^{w+t}}; \\
F^{\bar{\mu}\cdot 3^{t}}(x)\not\equiv x \pmod{3^{w+t+1}},
\end{cases}
\end{equation}
where $w\geq 2$, $t\geq 1$, $\bar{\mu}=\mu \bmod 3\in \{1, 2\}$.
\end{lemma}
\begin{proof}
Assume $\mu\bmod 3=1$, we prove this lemma via mathematical induction on $t$. According to relation~\eqref{eq:3:change},
one has $F(x) = x+k\cdot3^w$, where $k\bmod3\neq 0$.
Then one can calculate
$F^2(x)\equiv x+(1+\mu+2x)k\cdot3^w \not\equiv x\pmod {3^{w+2}}$ and
\begin{align*}
F^3(x)&\equiv F(x+(1+\mu+2x)k\cdot3^w) \pmod {3^{w+2}}\\
&\equiv F(x)+(2\mu x+\mu)(1+\mu+2x)k\cdot3^w \pmod {3^{w+2}}\\
&\equiv F(x)+(\mu +\mu^2)k\cdot3^w + (4\mu+2\mu^2)xk\cdot3^w\pmod {3^{w+2}}\\
&\equiv x+(1+\mu+\mu^2)k\cdot3^w \pmod {3^{w+2}}.
\end{align*}
As $\mu\bmod 3 =1$, it yields $(1+\mu+\mu^2)\bmod 9=3$, and
$(\sum^{3^s-1}_{j=0}\mu^j)\bmod 3=0$.
So, from the above congruence, one has
\begin{equation*}
\begin{cases}
F^{3}(x)\equiv x \pmod{3^{w+1}}; \\
F^{3}(x)\not\equiv x \pmod{3^{w+2}},
\end{cases}
\end{equation*}
and relation~\eqref{eq:3t:change} holds for $t=1$.
Suppose that relation~\eqref{eq:3t:change} holds for $t=s$, namely
\begin{equation}
\label{eq:f3s}
\begin{cases}
F^{3^s}(x)\equiv x \pmod{3^{w+s}}; \\
F^{3^s}(x)\not\equiv x \pmod{3^{w+s+1}}.
\end{cases}
\end{equation}
When $t=s+1$, setting $n=3^s$ in Property~\ref{property:f2_3t}, one can get
\begin{equation}\label{eq:temp:f3s}
F^{3\cdot 3^s}(x)\equiv x+(\mu^{2\cdot3^s}+\mu^{3^s}+1)k'\cdot 3^{w+s} \pmod {3^{w+s+2}}
\end{equation}
from relation~\eqref{eq:f3s}, where $k'\bmod3\neq0$.
It can be known
$(\mu^{2\cdot3^s}+\mu^{3^s}+1) \bmod 9 =3$
from $\mu\bmod 3=1$.
Thus, congruence~\eqref{eq:temp:f3s} becomes
\[F^{3^{s+1}}(x) \equiv x +k'\cdot3^{w+s+1}\pmod {3^{w+s+2}},\]
which yields that relation~\eqref{eq:3t:change} holds for $t=s+1$.
The above induction completes proof of the lemma when $\mu\bmod3=1$.

The proof for the case $\mu\bmod3=2$ is similar and omitted here.
\end{proof}

\begin{theorem}\label{thero:3n:period}
Given an initial value $x_0\in \mathbb{Z}_{3^n}$, the period of sequence
$S(x_0; \mu, 3^n)$ is
$$L(x_0; \mu, 3^n)=\bar{\mu}\cdot 3^{n-v_{x_0}}$$
when $n\ge v_{x_0}$;
$L(x_0; \mu, 3^n)\leq \bar{\mu}$ otherwise,
where $\bar{\mu}=\mu\bmod 3$,
\begin{equation}
\label{eq:def:vx0}
v_{x_0}=\max\{t \mid F^{\bar{\mu}}(x_{i^*})\equiv x_{i^*} \pmod {3^t}\},
\end{equation}
and
\begin{equation}\label{eq:3n:i}
i^*=
\begin{cases}
0& \mbox{if } x_0\bmod 3=0;\\
1& \mbox{if } x_0\bmod 3=2;\\
2& \mbox{if } x_0\bmod 3=1.\\
\end{cases}
\end{equation}
\end{theorem}
\begin{proof}
As proof of this theorem is similar for different value of $\bar{\mu}$, here
we only present the proof for the case $\bar{\mu}=2$.

According to Property~\ref{pro:F3n},
if $x_0\bmod 3\neq 0$, then $x_{i^*}\bmod3=0$ and sequence $S(x_{i^*}; \mu, 3^n)$ is periodic.
Thus, we only analyze the period of sequence $S(x_0; \mu, 3^n)$ for any $x_0\in \mathbb{H}_{3^n}$.
Referring to Eq.~\eqref{eq:def:vx0} and $\bar{\mu}=2$,
one has
\begin{equation}
\label{eq:3:vx0}
\begin{cases}
F^2(x_0)\equiv x_0 \pmod{3^{v_{x_0}}}; \\
F^2(x_0)\not\equiv x_0 \pmod{3^{v_{x_0}+1}}.
\end{cases}
\end{equation}
When $n<v_{x_0}$, combining relation~\eqref{eq:3:vx0} and the definition of $L(x_0; \mu, 3^n )$, one has $L(x_0; \mu, 3^n )\leq2$.
When $n\geq v_{x_0}$, one can prove
\begin{equation}\label{eq:th1:it}
L(x_0; \mu, 3^{v_{x_0}+t})=2\cdot 3^{t}
\end{equation}
by mathematical induction on $t$.
When $t=0$, $L(x_0; \mu, 3^{v_{x_0}})\leq 2=\bar{\mu}$ from relation~\eqref{eq:3:vx0}.
Assume $L(x_0; \mu, 3^{v_{x_0}}) =1$, then $F(x_0)\equiv x_0 \pmod {3^{v_{x_0}}}$, it means
$F(x_0)=x+k\cdot3^{v_{x_0}}$, where $k$ is an integer.
So
$$F^2(x_0)\equiv x_0+(1+\mu)k\cdot3^{v_{x_0}} \equiv x_0\pmod{3^{v_{x_0}+1}}.$$
But the above congruence contradicts with relation~\eqref{eq:3:vx0}, so $L(x_0; \mu, 3^{v_{x_0}})=2$ and
Eq.~\eqref{eq:th1:it} holds for $t=0$.

Suppose that Eq.~\eqref{eq:th1:it} holds for $t=s$, namely,
\begin{equation}\label{eq:temp1}
L(x_0; \mu, 3^{v_{x_0}+s}) =2\cdot 3^s.
\end{equation}
When $t=s+1$,
from the definition of $L(x_0; \mu, 3^{v_{x_0}+s+1})$, one has
$F^{L(x_0; \mu, 3^{v_{x_0}+s+1})}(x_0)\equiv x_0 \pmod{3^{v_{x_0}+s+1}}$.
Then
$F^{L(x_0; \mu, 3^{v_{x_0}+s+1})}(x_0)\equiv x_0 \pmod{3^{v_{x_0}+s}}$
and $L(x_0; \mu, 3^{v_{x_0}+s})$ divides
$L(x_0; \mu, 3^{v_{x_0}+s+1})$, which
further yields $2\cdot 3^s$ divides
$L(x_0; \mu, 3^{v_{x_0}+s+1})$ from Eq.~\eqref{eq:temp1}.
According to relation~\eqref{eq:3:vx0} and Lemma~\ref{lemma:3:change}, one has
\begin{equation*}
\begin{cases}
F^{2\cdot 3^{s+1}}(x_0)\equiv x_0 \pmod{3^{v_{x_0}+s+1}}; \\
F^{2\cdot 3^s}(x_0)\not\equiv x_0 \pmod{3^{v_{x_0}+s+1}}.
\end{cases}
\end{equation*}
It means that $L(x_0;\mu, 3^{v_{x_0}+s+1})\neq 2\cdot 3^s$ and $L(x_0; \mu$, $3^{v_{x_0}+s+1})$ divides $ 2\cdot3^{s+1}.$
Thus, $L(x_0; \mu, 3^{v_{x_0}+s+1}) = 2\cdot3^{s+1}$.
So, Eq.~\eqref{eq:th1:it} holds for $t=s+1$.
The above induction completes the proof of Eq.~\eqref{eq:th1:it}.
Setting $t=n-v_{x_0}$ in Eq.~\eqref{eq:th1:it} completes proof of the theorem for the typical case.
\end{proof}

When $\mu=20$, $x_0=50$, and $n=7$, one can detect
$L(50; 20, 3^7)=486$ via numerical simulation.
In comparison, one can calculate $\bar{\mu}=2$ and $v_{x_0}=2$,
which further produces $L(50; 20, 3^7)=2\cdot3^{7-2}=486$ from Theorem~\ref{thero:3n:period}.
According to Theorem~\ref{thero:3n:period},
Eq.~\eqref{eq:Lmax:3n} is revised and shown in Corollary~\ref{Cor:x_0:max3n}.
Moreover, referring to the proof process of Corollary~\ref{Cor:x_0:max3n}, one has $L(x_0; \mu,3^n)=L(\mu, 3^n)$
when
\[
f^{i^*}_{3^n}(x_0)\bmod 3^3 \in
\begin{cases}
A\cup B & \mbox{if } \mu\bmod 9\in\{1,2,5\};\\
A & \mbox{if } \mu\bmod 9=4;\\
B & \mbox{if } \mu\bmod 9=7,
\end{cases}
\]
where $i^*$ is given in Eq.~\eqref{eq:3n:i}, $A=\{3, 12, 21\}$, and $B=\{6, 15, 24\}$.

\begin{corollary}
\label{Cor:x_0:max3n}
The maximum period of sequences generated by iterating the Logistic map over ring $\Zy{n}$ is
\[
L(\mu, 3^n)=
\begin{cases}
1 & \mbox{if } \mu\bmod 3=0;\\
3^{n-2} & \mbox{if } \mu\bmod 3=1;\\
2\cdot 3^{n-2} & \mbox{if } \mu\bmod 9\in\{2,5\};\\
2\cdot 3^{n-3} & \mbox{if } \mu\bmod 9=8 \textit{~and~} \mu\bmod27 \neq 17;\\
2\cdot 3^{n-4} & \mbox{if } \mu\bmod27 = 17.
\end{cases}
\]
\end{corollary}

\begin{proof}
When $x_0\mod 3\neq 0$, $x_{i^*}= f^{i^*}_{3^{n}}(x_0)\equiv 0\pmod 3$ from Eq.~\eqref{eq:3n:i}.
So, the period analysis of sequence $S(x_0; \mu, 3^n)$ is the same no matter the value of $x_0\mod 3$.
In the following analysis, we assume $x_0=3k$, where $k$ is an integer.
Depending on the value of $\mu$, the proof is divided
into the following three cases:
\begin{itemize}
\item $\mu\bmod3=0$: one has $L(\mu, 3^n)=1$ from~\cite[Theorem 1]{Yang:log:Scis2017}.

\item $\mu\bmod3=1$: 
In such case, $F(x)\equiv x\pmod {3^2}$ for any $x\in \{0,3,6\}$.
It means that
\begin{equation}\label{eq:3x^2}
F(x_0)\equiv F(x_0\bmod 3^2)\equiv x_0 \pmod{3^2}
\end{equation}
for any $x_0\in \mathbb{H}_{3^n}$.
It yields $v_{x_0}\geq 2$.
Then $L(x_0; \mu, 3^n)\leq 3^{n-2}$ from Theorem~\ref{thero:3n:period}.
So, $L(\mu, 3^n) \leq 3^{n-2}$.
As
\[
F(3)\equiv
\begin{cases}
12 \pmod {3^3} & \mbox{if } \mu\bmod 9 =1;\\
21 \pmod {3^3} & \mbox{if } \mu\bmod 9 =4;\\
~3 \pmod {3^3} & \mbox{if } \mu\bmod 9 =7,
\end{cases}
\]
and
\[
F(6)\equiv
\begin{cases}
15 \pmod {3^3} & \mbox{if } \mu\bmod 9 =1;\\
~6 \pmod {3^3} & \mbox{if } \mu\bmod 9 =4;\\
24 \pmod {3^3} & \mbox{if } \mu\bmod 9 =7,
\end{cases}
\]
and congruence~\eqref{eq:3x^2}, one can get
$v_{x_0}=2$
if initial value $x_0$ satisfies
\[
x_0\bmod 3^3\in
\begin{cases}
\{3,12,6,15\} & \mbox{if } \mu\bmod 9 =1;\\
\{3,21\} & \mbox{if } \mu\bmod 9 =4;\\
\{6,24\} & \mbox{if } \mu\bmod 9 =7.
\end{cases}
\]
It yields from Theorem~\ref{thero:3n:period} that $L(x_0; \mu,3^n)=3^{n-2}$.
Therefore, $L(\mu, 3^n) =3^{n-2}$.

\item $\mu\bmod3=2$:
one can calculate
\[
F^2(3k)=\mu^3(3k)^4+2\mu^3(3k)^3+(\mu^3+\mu^2)(3k)^2+3k\mu^2.
\]
So,
\begin{equation}\label{eq:f(3k)}
F^2(3k)\equiv 3k \pmod{3^2},
\end{equation}
and
\begin{equation}\label{eq:F3k:m3}
\begin{aligned}
F^2(3k) &\equiv(\mu^3+\mu^2)(3k)^2+3k\mu^2\pmod{3^3}\\
&\equiv
\begin{cases}
12k \pmod{3^3}& \mbox{if } \mu\bmod 9=2;\\
21k \pmod{3^3} & \mbox{if } \mu\bmod 9=5;\\
3k \pmod{3^3} & \mbox{if } \mu\bmod 9=8,
\end{cases}
\end{aligned}
\end{equation}
and
\begin{equation}\label{eq:F3k:m4}
\begin{aligned}
F^2(3k) &\equiv 2\mu^3(3k)^3+(\mu^3+\mu^2)(3k)^2+3k\mu^2 \pmod{3^4}\\
&\equiv
\begin{cases}
(3k)^3+ 30k \pmod{3^4} & \mbox{if } \mu\bmod 27=8;\\
3k \pmod{3^4}          & \mbox{if } \mu\bmod 27=17;\\
(3k)^3+ 3k\pmod{3^4}   & \mbox{if } \mu\bmod 27=26.
\end{cases}
\end{aligned}
\end{equation}
As for any $x_0\in \mathbb{H}_{3^n}$, it yields from Eqs.~\eqref{eq:f(3k)},~\eqref{eq:F3k:m3},~\eqref{eq:F3k:m4} that
\begin{align*}
v_{x_0} \geq
\begin{cases}
2 & \mbox{if } \mu\bmod 9\in\{2,5\};\\
3 & \mbox{if } \mu\bmod 9=8 \textit{~and~} \mu\bmod27 \neq 17;\\
4 & \mbox{if } \mu\bmod27 = 17.
\end{cases}
\end{align*}
Thus, from Theorem~\ref{thero:3n:period}, one has
\begin{align*}
L(\mu, 3^n)\leq
\begin{cases}
2\cdot 3^{n-2} & \mbox{if } \mu\bmod 9\in\{2, 5\};\\
2\cdot 3^{n-3} & \mbox{if } \mu\bmod 9=8 \mbox{ and } \mu\bmod 27 \neq 17;\\
2\cdot 3^{n-4} & \mbox{if } \mu\bmod 27 = 17.
\end{cases}
\end{align*}
If the initial value satisfies
$x_0\bmod 3^t \in \{3,6\}$,
one can obtain $v_{x_0}=t-1$ by combining Eqs.~\eqref{eq:f(3k)},~\eqref{eq:F3k:m3}, and~\eqref{eq:F3k:m4},
where
\begin{align*}
t=
\begin{cases}
3   & \mbox{if } \mu\bmod 9\in\{2,5\};\\
4   & \mbox{if } \mu\bmod 9=8 \mbox{~and~} \mu\bmod27 \neq 17;\\
5   & \mbox{if } \mu\bmod 27 = 17.
\end{cases}
\end{align*}
It means $L(x_0; \mu, 3^n)=2\cdot3^{n-t+1}$.
Thus, $L(\mu, 3^n)=2\cdot3^{n-t+1}$ and this corollary holds if $\mu\bmod 3=2$.
\end{itemize}
\end{proof}

\section{Conclusion}

This paper presented explicit expression of the period of sequences generated by iterating the Logistic map from any initial value in ring $\Zy{n}$. Based on the explicit expression, we disclose the maximum period of the sequences.
Moreover, we present sufficient and necessary condition for the sequences achieving the maximum period.
The analysis method can be extended to the variants of the Logistic map and other chaotic maps over ring $\mathbb{Z}_{p^n}$.
Comparing with the chaotic maps owning bijective functional graph in a digital domain, say Cat map studied in 
\cite{cqli:Cat:TC22}, the functional graphs of the Logistic map are much more complex.
Much efforts are deserved to analyze their graph structure over various domains.

\section*{Acknowledgement}
This work was supported by the National Natural Science Foundation of China (no.~92267102) and Scientific
Research Fund of Hunan Provincial Education
Department under Grant 20C1759,
Postgraduate Scientific Research Innovation Project of Hunan Province (no.~CX20210601).
\bibliographystyle{ws-ijbc}
\bibliography{net_log}
\end{document}